\documentclass[a4paper,UKenglish,cleveref, autoref, thm-restate]{lipics-v2021}

\usepackage{amsmath, amsthm, amssymb}

\usepackage{bm}
\usepackage{cite}
\usepackage{nicefrac}
\usepackage{graphicx}
\usepackage{tikz}
\usepackage{todonotes}
\usepackage{hyperref} 
\usetikzlibrary{arrows,arrows.meta,bending,shapes,calc,intersections,positioning}
\tikzstyle{bank}=[circle, draw, inner sep = 1]

\usepackage[ruled,vlined,linesnumbered]{algorithm2e}

\newcommand{\N}{\mathbb{N}}

\newcommand{\veca}{\mathbf{a}}

\newcommand{\vecp}{\mathbf{p}}
\newcommand{\vecr}{\mathbf{r}}

\newcommand{\vecBeta}{\bm{\beta}}

\newcommand{\vecEll}{\bm{\ell}}

\newcommand{\calF}{\mathcal{F}}

\newcommand{\classNP}{\textsf{NP}}

\newcommand{\MH}[1]{{\color{orange} Martin: #1}}

\allowdisplaybreaks

\title{Fractional Claims Trades and Donations in Financial Networks}

\titlerunning{Fractional Claims Trades and Donations} 

\author{Martin Hoefer}{RWTH Aachen University, Germany }{mhoefer@em.uni-frankfurt.de}{https://orcid.org/0000-0003-0131-5605}{}

\author{Lars Huth}{RWTH Aachen University, Germany }{huth@algo.rwth-aachen.de}{https://orcid.org/0009-0008-6855-8525}{}

\author{Lisa Wilhelmi}{RWTH Aachen University, Germany }{wilhelmi@algo.rwth-aachen.de}{https://orcid.org/0000-0003-0845-1941}{}

\authorrunning{M. Hoefer, L. Huth, L. Wilhelmi} 

\Copyright{Martin Hoefer, Lars Huth, Lisa Wilhelmi} 

\ccsdesc[100]{} 

\keywords{Financial Networks, Claims Trade, Donations} 

\relatedversion{} 

\funding{Supported by DFG Research Unit ADYN (project number 411362735).}

\ArticleNo{0}

\nolinenumbers 

\renewcommand{\paragraph}[1]{\medskip \noindent \textsf{\textbf{#1}} $\,$}

\begin{document}

\maketitle

\begin{abstract}
    Exploring measures to improve financial networks and mitigate systemic risks is an ongoing challenge. We study claims trading, a notion defined in Chapter 11 of the U.S. Bankruptcy Code. For a bank $v$ in distress and a trading partner $w$, the latter is taking over some claims of $v$ and in return giving liquidity to $v$. The idea is to rescue $v$ (or mitigate contagion effects from $v$’s insolvency). We focus on the impact of trading claims fractionally, when $v$ and $w$ can agree to trade only part of a claim. In addition, we study donations, in which $w$ only provides liquidity to $v$. They can be seen as special claims trades.

    When trading a single claim or making a single donation in networks without default cost, we show that it is impossible to strictly improve the assets of both banks $v$ and $w$. Since the goal is to rescue $v$ in distress, we study creditor-positive trades, in which $v$ improves and $w$ remains indifferent. We show that an optimal creditor-positive trade that maximizes the assets of $v$ can be computed in polynomial time. It also yields a (weak) Pareto-improvement for all banks in the entire network. In networks with default cost, we obtain a trade in polynomial time that weakly Pareto-improves all assets over the ones resulting from the optimal creditor-positive trade. We generalize these results to trading multiple claims for which $v$ is the creditor.

    Instead, when trading claims with a common debtor $u$, we obtain \classNP-hardness results for computing trades in networks with default cost that maximize the assets of the creditors and Pareto-improve the assets in the network. Similar results apply when $w$ donates to multiple banks in networks with default costs. For networks without default cost, we give an efficient algorithm to compute optimal donations to multiple banks.
\end{abstract}

\section{Introduction}
Systemic risk is an important threat in financial networks. Since the global 2007/08 banking crises there have been repeated incidents that exemplified the potentially disastrous consequences of default and contagion in the financial system. Recently, the challenges of stability in financial markets have become apparent with problems surrounding Credit Suisse in 2023 and the subsequent acquisition by UBS. The resolution was orchestrated by the Swiss Federal Council. It enabled (and essentially forced through) the acquisition by UBS to stabilize the global financial markets. It can be seen as a mixture of acquisition and government bailout. Such bailouts have also been applied in earlier times of crises, but there is often a substantial discomfort in the society to use taxpayer money for bailouts of banks. As such, alternative, more intrinsic measures to stabilize financial networks are needed.

A growing body of research strives to design and understand measures to mitigate the effects of default in financial networks. Tools and techniques from theoretical computer science are used to characterize the inherent complexity of financial networks and to design algorithms that improve their stability. A very useful approach to capture key features of financial networks is the model by Eisenberg~and~Noe~\cite{EisenbergN01}. There are $n$ banks represented by nodes in a directed graph. Directed edges represent debt claims, and edge weights are liabilities. Each bank can use incoming payments to clear its debt. Moreover, each bank has a non-negative amount of external assets, which can also be used for payment of claims. Every bank pays their debts in proportion to the liability. Fundamental algorithmic work in this model focused on the computation of \emph{clearing states}, i.e., fixed-point solutions for payments in the network. These states are known to measure the exposure of banks in the network to insolvencies. Recently, there is work to understand the effects of changing external assets, payment rules, and/or network structure, which we discuss in detail below. 

In this paper, we consider the operation of \emph{claims trading} to rescue a bank in distress. Claims trading is defined in Chapter 11 of the U.S.\ Bankruptcy Code. The idea was formalized recently in the context of the Eisenberg-Noe model in~\cite{HoeferVW24}, and it represents a local change of network structure and external assets. Suppose a bank $v$ is in distress (e.g., Credit Suisse) and there is a potential trading partner $w$ (e.g., UBS or a central bank). Bank $w$ is taking over some claims of $v$ and in return giving liquidity to $v$. The idea is to rescue $v$, or if that is impossible, at least to mitigate the contagion effects from $v$’s insolvency. As such, in claims trading the bank $w$ buys claims from $v$ and provides immediate funds in return. This avoids a centralized intervention of forcing another bank into an acquisition or using guarantees by the government for a bailout. Instead, buyer $w$ shall recover the cash invested in buying a claim via increased payments from within the network. 

Claims trading was formalized and studied in~\cite{HoeferVW24}, but only with \emph{binary} trades -- claims can either be traded entirely or not at all. We consider a substantial generalization of this operation as, more realistically, a claim is not a rigid body and can usually be split or paid partially from different sources.
As such, the more natural scenario is \emph{fractional} claims trading, where $v$ and $w$ can agree to trade only parts of a claim. In this paper, we initiate the study of algorithms for fractional claims trading. In addition, we incorporate default costs and study a natural operation of \emph{donation} in which $w$ provides liquidity to $v$ without trading any of the claims. Again, we incorporate incentives in the sense that donor $w$ shall recover the funds invested via increased payments within the network. As a preliminary insight, we observe that donations can be interpreted as a special case of (binary) claims trades. 

\paragraph{Contribution and Outline.}
We study fractional claims trading in the financial network model by Eisenberg~and~Noe~\cite{EisenbergN01}, as well as in a generalization of that model with default cost~\cite{RogersV13}. Since a bank $v$ has much more control about claims in which it acts as creditor, we focus on trading claims in which $v$ is the creditor (i.e., \emph{incoming} edges in the network). 

After a formal definition of the model in Section~\ref{sec:model}, our first observation in Section~\ref{sec:prelim} is that in networks without default cost, any single claims trade or single donation can only lead to strict improvement for \emph{either} the assets of creditor/recipient $v$ \emph{or} buyer/donor $w$. Since our goal is to rescue $v$, we focus on \emph{creditor-positive} trades, which generate a strict improvement for the creditor and do not deteriorate the assets for the buyer. These trades also yield a Pareto-improvement of all banks in the network. Our interest lies in \emph{optimal} creditor-positive trades that maximize the assets\footnote{Claims trades leave the liabilities of each bank fixed. By optimizing the assets we also optimize the equity of the creditor.} of the creditor, and as a special case, optimal donations that maximize the assets of the recipient. They also maximize the Pareto-improvement for all banks.

In Section~\ref{sec:single} we focus on trading a single claim and, as a special case, making a single donation. In networks without default cost, we show that an optimal creditor-positive trade that maximizes the assets of $v$ can be computed in polynomial time. More generally, for arbitrary default cost, we obtain a trade in polynomial time that weakly Pareto-improves the assets of every bank over the ones resulting from the optimal creditor-positive trade. Our algorithm computes \emph{some} trade that provides, pointwise for every bank, at least as many assets and equity as the optimal creditor-positive trade. In networks with default cost, it could also be a \emph{positive} one, in which both $v$ and $w$ strictly improve. This result relies on a \emph{default hierarchy}, our main technical tool that allows us to handle sets of insolvent banks in Section~\ref{sec:hierarchy}. At the basis of our approach we analyze the clearing states that emerge after a trade based on the assets of creditor $v$. The default hierarchy decomposes the possible assets of $v$ into up to $n$ intervals, which correspond to regions in which the set of insolvent banks remains invariant. When the assets of $v$ transition from a higher to an adjacent lower interval, a non-empty set of additional banks slide into default. We give a polynomial-time algorithm to find the intervals and the corresponding sets of insolvent banks. Our algorithm uses a procedure for clearing state computation in~\cite{RogersV13} as a subroutine, but needs to overcome further technical challenges in identifying the correct (possibly open) interval borders along with corresponding sets of insolvent banks. Here we leverage an intricate connection between \emph{minimization} of the assets of $v$ at an interval border and the clearing state being a \emph{maximal} fixed point. Given the default hierarchy, we formulate computation of the trade as a set of non-linear mathematical programs in Section~\ref{sec:trade}. Interestingly, using a suitable variable substitution we can transform them into LPs that can be solved in polynomial time.

In Section~\ref{sec:incoming} we generalize the results for single trades to trading multiple claims for which $v$ is the common creditor. This generalization relies on a characterization that optimal trades of multiple claims require to trade at most one claim fractionally. This allows to use our approach for trading single claims as a building block. The fractional claim is determined based on the sorted order of payments of the debtors. We refine the default hierarchy to also express the ordering of claims in terms of payments. On each level of the hierarchy the payments are linear in the assets of $v$, so there are only a polynomial number of orderings.


Instead of increasing the assets or even rescuing a bank $u$, in Section~\ref{sec:outgoing} we try to mitigate the contagion of $u$'s default by trading claims with the creditors of $u$. We again focus on trades that Pareto-improve the network and maximize the sum of assets of the creditors of $u$. Here, we obtain \classNP-hardness for computing the trade which maximizes the assets of $u$'s creditors in networks with default costs. A similar result applies when $w$ donates to multiple banks in networks with default costs. In contrast, for networks without default cost we can provide an efficient algorithm to compute optimal donations to multiple banks.

\paragraph{Related Work.}
Our work is based on the network model in~\cite{EisenbergN01} and the extension to default cost discussed in~\cite{RogersV13}. The basic solution concept in this model is a clearing state, a fixed-point assignment of payments on each edge and resulting assets of banks. Existence, uniqueness, and computational complexity of clearing states in the standard model and in extensions to credit-default swaps has been of interest recently in~\cite{EisenbergN01,CsokaH24,SchuldenzuckerS17,IoannidisKV22,IoannidisKV23PPAD,SchuldenzuckerS20ambiguity,PappW21}.

A notable feature of the standard model is the assumption of proportional payment. This assumption has been changed and generalized in several works, which shed light on structure and computation of clearing states with ranking-based or even general monotone payments as well as game-theoretic incentives of banks to change them~\cite{BertschingerHS20,CsokaH18,KanellopoulosKZ21,HoeferW22,IoannidisKV23}.

Injecting additional assets into the network and bailouts have been studied in static networks~\cite{KanellopoulosKZ22,PapachristouK22} and networks that are subject to repeated shocks~\cite{PapachristouBK23}. In contrast, all transfers of assets in claims trading are intrinsic to the network, and a bank providing these assets must not be harmed. 

Several recent studies have proposed network adjustments to increase the assets for some or even all banks. Related to our work are \emph{debt swaps}, in which two creditors swap debt contracts with the same liability. Indeed, a single claims trade can be interpreted as a debt swap with an artificial auxiliary claim. The structural and computational properties of debt swaps have been analyzed in~\cite{PappW21EC}. Computational complexity of sequential debt swaps in networks with general monotone payments were analyzed in~\cite{FroeseHW23}. Notably, this is different from approaches that study sequential implementation of payments in a static network~\cite{PappW21sequential,CsokaH18}.

Other network adjustments were analyzed in game-theoretic scenarios, including operations of forgiving or canceling debt~\cite{PappW20,KanellopoulosKZ22,KeijzerTV24} (in which a creditor strategically lowers an incoming claim) or debt transfers~\cite{KanellopoulosKZ23} (in which a claim with creditor $v$ is forward directly to one of the creditors of $v$). Here pure equilibria are often absent, and their existence is hard to decide. As a notable contrast, lending games~\cite{egressy2024price} yield a Bertrand-style price competition for lending assets to a given bank (which then introduces new liability relations into the network) and give rise to unique equilibrium payments.

Most closely related to our paper is work on donations~\cite{TongKV24} or pre-payments~\cite{prepayments}. The focus is incentives and equilibrium existence in a game-theoretic setting (that allows restricted classes of donations) rather than the direct optimization problems that we solve here.
Our work generalizes the model of (binary) claims trading proposed recently in~\cite{HoeferVW24}. The authors consider networks without default cost and general monotone payment functions. The paper shows \classNP-hardness results and algorithms that guarantee approximate solutions. Our work considers fractional trades, for which we encounter different conditions. Notably, optimal binary multi-trades of incoming edges are \classNP-hard -- for fractional ones we give an efficient algorithm, even in networks with default cost.

\section{Model and Preliminaries}
\subsection{Claims Trades and Donations}
\label{sec:model}

\paragraph{Network Model.} 
We consider claims trades in the seminal network model of Eisenberg and Noe~\cite{EisenbergN01}. A financial network $\calF = (V,E,\vecEll,\veca^x)$ consists of (1) a set $V$ of financial institutions (termed \emph{banks} throughout), (2) a set of edges $E$, where edge $(u,v) \in E$ is used to model a \emph{debt relation} between creditor $v \in V$ and debtor $u \in V$, (3) a vector $\vecEll = (\ell_e)_{e \in E}$ of \emph{liabilities}, i.e., $\ell_{(u,v)} \in \N_{>0}$ is an edge-weight\footnote{For convenience, all liabilities and external assets are represented as integers in standard logarithmic encoding. All our results hold similarly when liabilities and external assets are rational numbers.} and specifies the amount that $u$ owes to $v$, and (4) a vector $\veca^x = (a_v^x)_{v \in V}$ of \emph{external assets}\footnote{We use the superscript $x$ to indicate that these are assets external to the network.}, i.e., $a_v^x \in \N_{\ge 0}$ is the amount of money owned by bank $v$ that is not inherent to the network. We use $n = |V|$ to refer to the number of banks. The set of outgoing edges of a bank $v$ is denoted by $E^+(v) = \{(v,u) \in E \mid u \in V\}$ and the set of incoming ones by $E^-(v) = \{(u,v) \in E \mid u \in V\}$. The \emph{total liabilities} $L_v$ of $v$ is the sum of weights of outgoing edges of $v$, i.e., $L_v = \sum_{e \in E^+(v)} \ell_e$.

We use $a_v$ to denote the total assets available to bank $v$, which are external assets as well as funds received over incoming edges. Bank $v$ is assumed to distribute her total assets proportionally to its outgoing edges, i.e., edge $e \in E^+(v)$ receives a \emph{payment} of $p_e = \min\{a_v, L_v\} \cdot \ell_e/L_v$. If $a_v \ge L_v$, then the bank is solvent. Any excess funds $a_v - L_v$ remain as a profit with $v$, and $v$ does not overpay its claims. Otherwise, if $a_v < L_v$, the bank is in default. We consider a generalized variant of the model, which includes a notion of \emph{default cost}~\cite{RogersV13} -- in case of default, the assets of bank $v$ are further decreased to a fraction $\delta \in [0,1]$. The total assets of $v$ are 
\[
    a_v = \begin{cases} a_v^x + \sum_{e \in E^-(v)} p_e & \text{if $v$ is solvent, or}\\
                        \delta \cdot  \left(a_v^x + \sum_{e \in E^-(v)} p_e\right) & \text{otherwise.}\end{cases}
\]
These equalities, together with non-negativity $p_e \ge 0$, yield a set of fixed-point constraints for the vector $\vecp = (p_e)_{e \in E}$ of payments. The set of all fixed-point payments is known to form a complete lattice~\cite{EisenbergN01,RogersV13}. Indeed, the fixed point is often unique~\cite{EisenbergN01,CsokaH24}. In any case, we follow standard conventions in the literature and focus on the \emph{supremum} of the lattice, which we term the \emph{clearing state} $\vecp$ of $\calF$. It can be computed in polynomial time~\cite{EisenbergN01}. Given the clearing state $\vecp$, we obtain the \emph{total assets} $a_v$ given above and the \emph{recovery rate} $r_v = \min\{a_v/L_v, 1\}$. The clearing state can then be expressed by $p_e = r_u \cdot \ell_e$ for each $u \in V$ and $e = (u,v) \in E$. 

\paragraph{Claims Trades.} 
If a bank $v$ in default is unable to settle all its debt, this can represent a risk for the network. The creditors of $v$ then receive less payments, which can lead to further defaults, and in this way a dynamics of defaults can spread in the network. Similar to recent work~\cite{HoeferVW24}, we consider a network adjustment to address such potentially disastrous defaults: A benevolent bank (such as, say, a central bank) can buy incoming edges of $v$, thereby take over the risk of default of $v$ and decreased payments on the edges. In turn, $w$ ensures $v$ is solvent by providing guaranteed funds directly. Our focus in this paper is a general class of \emph{fractional} trades. Instead of a \emph{binary} decision to transfer an \emph{entire} edge $(u,v)$ or not~\cite{HoeferVW24}, we allow to transfer an arbitrary fraction $\beta \in [0,1]$ of the edge to $w$.

Formally, consider banks $u,v$ and $w$ with edge $e$ with $e = (u,v)$ and $\ell_{e} \ge 0$. Suppose $u$ is in default. For a claims trade with traded fraction $\beta \in [0,1]$, we change $\calF$ to $\calF' = (V,E \cup \{(u,w)\},\vecEll',\veca'^x)$, the \emph{post-trade network}, as follows. We add a new edge\footnote{If $(u,w)$ already exists in $E$, we add $\beta\ell_e$ to $\ell_{(u,w)}$ and set $\ell'_{(u,w)} = \ell_{(u,w)} + \beta \ell_e$. Due to proportional payments, this has the same effect as creating a new, separate multi-edge for the traded claim.} $(u,w)$ with liability $\ell'_{(u,w)} = \beta\ell_e$ and decrease the liability of $e$ to $\ell'_e = (1-\beta)\ell_e$. For all other edges $e' \in E \setminus \{(u,v),(u,w)\}$, there is no change $\ell'_{e'} = \ell_{e'}$. Thus, any payment from $u$ towards the claim will now be split in portions $1-\beta$ and $\beta$ received by $v$ and $w$, respectively. 

In turn, $w$ gives a \emph{return} payment $\rho \in [0,\beta \ell_e]$ to $v$. For convenience, we define $\alpha = \rho/(\beta \ell_e)$ as the \emph{haircut rate} and express the return by $\rho = \alpha \beta \ell_e$. Since $w$ is supposed to provide guaranteed funds to support the solvency of $v$, we separate the return from the clearing state and model it as a \emph{transfer of external assets} from $w$ to $v$. Hence, $w$ can invest at most her external assets as a return, so every trade must satisfy $\rho = \alpha \beta \ell_e \le a_w^x$. After a trade the external assets of $v$ and $w$ are $a'^x_v = a_v^x + \rho$ and $a'^x_w = a_w^x - \rho$. 
For all other banks the external assets remain unchanged, $a'^x_i = a^x_i$, for every $i \in V \setminus \{v,w\}$.

\paragraph{Variants of Claims Trades.}
We consider three variants of claims trades. For a \emph{single trade}, we are given a network $\calF$ with distinct banks $u,v,w \in V$ and an edge $e = (u,v) \in E$. 
The goal is to choose a fraction $\beta$ and a haircut rate $\alpha$. The clearing state in the post-trade network $\calF'$ is denoted by $\vecp'$. $v$ is called the \emph{creditor} and $w$ the \emph{buyer}. 

We consider two generalizations of single trades to multiple edges. For a \emph{multi-trade of incoming edges}, there are distinct banks $v,w$ in a network $\calF$. Consider the set of incoming edges of $v$ from banks other than $w$, i.e., $E^-(v) \setminus \{(w,v)\} = \{e_1,\ldots,e_k\}$. We can choose, for each $i \in [k]$, a fraction $\beta_i \in [0,1]$ and a haircut rate $\alpha_i \in [0,1]$. 
After the trade, a network $\calF'$ emerges: We change the liabilities exactly as in a sequence of single claims trades for each edge $e_i$. The external assets change accordingly, $a'^x_{v} = a^x_v + \sum_{i=1}^k \alpha_i \beta_i \ell_{e_i}$, and $a'^x_w = a^x_w - \sum_{i=1}^k \alpha_i \beta_i \ell_{e_i} \ge 0$. 

Similarly, for a \emph{multi-trade of outgoing edges}, there are distinct banks $u,w$ in a network $\calF$. Consider the set of outgoing edges of $u$ to banks other than $w$, i.e., $E^+(u) \setminus \{(u,w)\} = \{e_1,\ldots,e_k\}$. The goal is to choose, for each $i \in [k]$, a fraction $\beta_i \in [0,1]$ and a haircut rate $\alpha_i \in [0,1]$. 
After the trade, a network $\calF'$ emerges: The liabilities are the same as they were after performing a sequence of single trades of the edges $e_i$. The external assets change accordingly, $a'^x_{v_i} = a^x_{v_i} + \alpha_i \beta_i \ell_{e_i}$, and $a'^x_w = a^x_w - \sum_{i=1}^k \alpha_i \beta_i \ell_{e_i} \ge 0$. 

We proceed with a small example of trading a single claim.

\paragraph{Example.}
    Consider the simple example network $\calF$ in Figure~\ref{fig:ex-intro} (left). For simplicity, we assume no default cost, i.e., $\delta = 1$. The assets are $a_v = 2$ and $a_w = 5$. Suppose we trade the edge $(u,v)$ to $w$. Recall that our goal is to improve the conditions for the creditor bank $v$. By increasing the traded fraction $\beta$, the upper bound $\beta \ell_e$ for the return payments $\rho$ increases, and the remaining payments from $u$ to $v$ decrease. Hence, as $\beta$ grows, $w$ can pay a higher return -- and potentially has to do so in order to generate an improvement of the assets of $v$. 

    Consider the middle network $\calF'$. It emerges from the binary trade that obtains maximal assets of $v$. Since all liabilities are 4, with a return $\rho = a_w^x = 3$ we transfer all external assets to $v$, since they eventually get paid back to $w$. We obtain $a'_v = 3$ and $a'_w = a_w = 5$. 

    The right network $\calF'$ emerges from the fractional trade that obtains maximal assets of $v$. We only trade a $\frac 34$-fraction of the edge $(u,v)$. Hence, edge $(u,w)$ has liability $\beta \ell_e = 3$, which still allows a return of $\rho = a_w^x = 3$. Moreover, some payments of $u$ still reach $v$ over the remaining fraction of $(u,v)$ with liability 1. Overall, we obtain $a'_v = 3.5$ and $a'_w = a_w = 5$.

    This shows that fractional trades can yield strictly better assets than binary ones. Note that only one of $v$ or $w$ (but not both) improves strictly, and we have $\beta = 1$ and/or $\rho = a_w^x$. We prove these conditions in general, even for multi-trades, in Propositions~\ref{prop:noPositive}~and~\ref{prop:structure}. \hfill $\blacksquare$
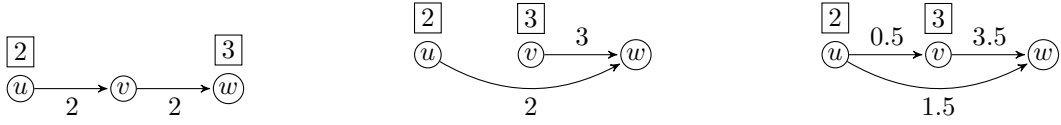
\begin{figure}[t]
    \centering
    \begin{tikzpicture}[>=stealth', shorten >=1pt, auto, node distance=1cm, transform shape, align=center]
        \node[bank] (u) at (0,0) {$u$};
        \node[bank] (v) [right=of u] {$v$};
        \node[bank] (w) [right=of v] {$w$};
        \node[rectangle, draw=black, inner sep=2.5pt] (xu) [above=0.1cm of u] {$2$};
        \node[rectangle, draw=black, inner sep=2.5pt] (xw) [above=0.1cm of w] {$3$};
        \draw[->] (u) -- node[below] {$2$} (v);
        \draw[->] (v) -- node[below] {$2$} (w);
    \end{tikzpicture}
    \hspace{2cm}
    \begin{tikzpicture}[>=stealth', shorten >=1pt, auto, node distance=1cm, transform shape, align=center]
        \node[bank] (u) at (0,0) {$u$};
        \node[bank] (v) [right=of u] {$v$};
        \node[bank] (w) [right=of v] {$w$};
        \node[rectangle, draw=black, inner sep=2.5pt] (xu) [above=0.1cm of u] {$2$};
        \node[rectangle, draw=black, inner sep=2.5pt] (xv) [above=0.1cm of v] {$3$};
        \draw[->] (u) to [bend right] node[below] {$2$} (w);
        \draw[->] (v) -- node[above] {$3$} (w);
    \end{tikzpicture}
    \hspace{2cm}
    \begin{tikzpicture}[>=stealth', shorten >=1pt, auto, node distance=1cm, transform shape, align=center]
        \node[bank] (u) at (0,0) {$u$};
        \node[bank] (v) [right=of u] {$v$};
        \node[bank] (w) [right=of v] {$w$};
        \node[rectangle, draw=black, inner sep=2.5pt] (xu) [above=0.1cm of u] {$2$};
        \node[rectangle, draw=black, inner sep=2.5pt] (xv) [above=0.1cm of v] {$3$};
        \draw[->] (u) -- node[above] {$0.5$} (v);
        \draw[->] (u) to [bend right] node[below] {$1.5$} (w);
        \draw[->] (v) -- node[above] {$3.5$} (w);
    \end{tikzpicture}
    \caption{    
    \label{fig:ex-intro}
    Left: Network $\calF$ from our introductory example. Boxed node labels indicate external assets. Edge labels indicate payments. All edges have a liability of 4. Middle: Network $\calF'$ after binary trade with $\beta = 1$ and $\rho = a_w^x = 3$. Right: Network $\calF'$ after fractional trade with $\beta = \nicefrac 34$ and $\rho = a_w^x = 3$.}
\end{figure}

\paragraph{Donations.}
A \emph{single donation} represents a transfer of (some of) the external assets $a_w^x$ of bank $w$ to bank $v$, without the trade of any claim. For convenience, we represent a donation by an equivalent claims trade as follows. We add (1) an auxiliary source bank $u$ without incoming edges and external assets, and (2) an edge $(u,v)$ with infinite liability $\ell_{(u,v)} = \infty$. Then any donation from $w$ to $v$ is equivalent to a claims trade of edge $(u,v)$ to $w$. In particular, since $u$ has no external assets and no incoming edges, the payment on $(u,v)$ is 0 before and after the trade. Moving edge $(u,v)$ to $w$ has no effect on the clearing state -- only the return $\rho$ (representing the donation) changes payments and assets. Since $\ell_{(u,v)} = \infty$, $\rho$ is only restricted by $a_w^x$.

In a \emph{multi-donation}, we are given a bank $w$ and a set $C \subseteq V \setminus \{w\}$ of banks. $w$ can make non-negative transfers to each of the banks in $C$, which sum up to at most the external assets $a_w^x$. By using a single auxiliary bank $u$ and an edge $(u,v)$ with infinite liability for each $v \in C$, we represent a multi-donation as a multi-trade of outgoing edges.

\subsection{Structural Preliminaries}
\label{sec:prelim}

We first observe that in networks without default cost, there exists no \emph{positive} single trade in which \emph{both $v$ and $w$ strictly} profit (i.e., with $a'_v > a_v$ \emph{and} $a'_w > a_w$). More generally, there is no positive multi-trade of incoming edges. Note that this result extends to single donations. 
\begin{proposition}
    \label{prop:noPositive}
    In networks without default cost, there is no positive multi-trade of incoming edges.
\end{proposition}
\begin{proof}
    A similar result holds for binary claims trades in~\cite{HoeferVW24}. The result follows from a reduction to an operation called debt swap studied in~\cite{PappW21EC,FroeseHW23}. Intuitively, if a positive trade (or debt swap) existed, we would reach a contradiction to the maximality of the clearing state in the pre-trade network $\calF$.

    We extend the result for binary trades to fractional ones. Suppose we have decided on a fraction $\beta$ for the trade. Due to proportional payments, we can replace $e$ by two multi-edges $e_1$ and $e_2$ between $u$ and $v$ with liability $\beta\ell_e$ and $(1-\beta)\ell_e$, respectively. This adjustment yields an equivalent network and clearing state, as well as the same assets of all banks. The fractional trade of $e$ in $\calF$ is a binary trade of $e_1$ in the adjusted network. Hence, the trade cannot be positive due to the result for binary trades. The argument extends to multi-trades of incoming edges by applying the network adjustment individually to all edges.
\end{proof}

Proposition~\ref{prop:noPositive} provides a formal justification to focus on trades, in which only one of the banks $v$ or $w$ does profit strictly. Recall that our motivation is to rescue a creditor $v$ in distress. As such, we focus on \emph{creditor-positive} trades, in which $a'_v > a_v$ and $a'_w = a_w$. Every creditor-positive trade weakly Pareto-improves the assets of \emph{all} banks in the network, as well as a strict improvement for $v$. This property is similar to binary trades and debt swaps~\cite{HoeferVW24,FroeseHW23,PappW21EC}. Therefore, an \emph{optimal} creditor-positive trade, which has maximal assets $a'_v$, also pointwise maximizes the assets of all banks. When trading incoming edges of $v$, our objective is to find a trade that gives every bank at least the post-trade assets of an optimal creditor-positive trade. 

When trading outgoing edges of $u$, we also strive to achieve a weak Pareto-improvement for all banks in the network and a strict improvement for at least one creditor $v_i$. Towards this end, we consider \emph{Pareto-positive} trades that guarantee $a'_w \ge a_w$, $a'_{v_i} \ge a_{v_i}$ for all $i \in [k]$, and $a'_{v_i} > a_{v_i}$ for at least one $i \in [k]$. Our objective is to find a Pareto-positive trade that maximizes the total assets of creditors $\sum_{i=1}^k a'_{v_i}$. Similarly, for multi-donations we focus on Pareto-positive donations, which achieve a weak Pareto-improvement for all banks in the network. The objective is to find a Pareto-positive multi-donation that maximizes the total assets of all banks $\sum_{b \in V} a'_b$.

In general, we can restrict attention to trades that satisfy further conditions. The next Proposition shows an equivalence that applies even in networks with default costs.

\begin{proposition}
    \label{prop:structure}
    For every fractional multi-trade of incoming edges, there is an equivalent trade for which $\rho = a_w^x$ or $\beta_i = 1$ for all $i \in [k]$. For every fractional multi-trade of outgoing edges, there is an equivalent trade for which $\sum_{i=1}^k \rho_i = a_w^x$ or $\beta_i = 1$ for all $i \in [k]$.
\end{proposition}
\begin{proof}
    First, consider a single trade with fraction $\beta$ and return $\rho$. If both $\beta < 1$ and $\rho < a_w^x$, we can increase the traded portion to $\hat{\beta} = \beta + \eta$ and the return to $\hat{\rho} = \rho + r'_u \cdot \eta$, where $r'_u$ is the post-trade recovery rate of $u$. Then $w$ increases the traded share of the edge by $\eta$. Thereby, we reroute an additional payment of $r'_u\cdot \eta$ of $u$ from $v$ to $w$. These additional payments are forwarded back to $v$ by increasing the return. Thus, the assets of $v$ and $w$ after a trade with $\hat{\beta}$ and $\hat{\rho}$ are the same as the ones after the original trade with $\beta$ and $\rho$. Moreover, the assets and recovery rates of all other banks in the network (including $u$) remain the same. By choosing $\eta = \min\{1-\beta, (a_w^x-\rho)/r'_u\}$, we obtain $\hat{\beta} = 1$ or $\hat{\rho} = a_w^x$.   

    The previous adjustment can be executed similarly for multi-trades of incoming or outgoing edges. We increase the fractions of the edges traded to $w$. Simultaneously, we increase $w$'s return to the creditor by the payments on the additional traded fraction of the edge, respectively. The net effect on all assets is 0. The adjustment terminates once we trade all edges fully (i.e., $\beta_i = 1$ for all $i \in [k]$) or the total return paid by $w$ is $a_w^x$.
\end{proof}

\section{Single Trades}
\label{sec:single}
For a given single trade consisting of an edge $e$ from $u$ to $v$ together with a buyer $w$, our goal is to compute a traded fraction $\beta$ and a haircut rate $\alpha$ that yield post-trade assets $a'_v$ for $v$ at least as high as the ones resulting from any creditor-positive trade.

\subsection{Computing a Default Hierarchy}
\label{sec:hierarchy}

Our approach is to parameterize the problem by the post-trade recovery rate $r'_v$. In particular, since we consider creditor-positive trades, we assume that $a'_w = a_w$ and, thus, $r'_w = r_w$. It will be useful to disentangle the recovery rates of $v$ and $w$ from their incoming assets. Formally, we consider the split network $\calF_s(v,w)$.
\begin{definition}
    Given a financial network $\calF$ and a subset of nodes $b_1,\ldots,b_k \subseteq V$, the \emph{split network} $\calF_s(b_1,\ldots,b_k)$ is given by replacing each bank $b_i$ by a sink node $b_{i,{in}}$ and a source node $b_{i,{out}}$. For every incoming (outgoing) edge $(b,b_i) \in E^-(b_i)$ ($(b_i,b) \in E^+(b_i)$) we introduce an incoming (outgoing) edge $(b,b_{i,{in}})$ of $b_{i,{in}}$ ($(b_{i,{out}},b)$ of $b_{i,{out}}$) with the same liability, respectively.
\end{definition}
Suppose that in $\calF_s(v,w)$ we set external assets $a^x_{v_{out}} = a_v$ and $a^x_{w_{out}} = a_w$, then the clearing state becomes equivalent to $\vecp$ in $\calF$. Moreover, any trade of any incoming edge of $v$ would only change incoming edges of $v_{in}$ and $w_{in}$ and, thus, the assets of these nodes. Now consider the vector of post-trade assets $\veca'$ and the clearing state $\vecp'$ resulting from any creditor-positive trade. Suppose that in $\calF_s(v,w)$ we assign $a^x_{v_{out}} = a'_v$ and $a^x_{w_{out}} = a'_w$ to $v_{out}$ and $w_{out}$. Then the clearing state becomes equivalent to $\vecp'$, except for the assets of $v_{in}$ and $w_{in}$. In particular, the clearing state in $\calF_s(v,w)$ will yield the same set of solvent banks as $\vecp'$.

Based on this observation, we decompose our search based on the hierarchy of insolvent banks. For computation we apply a search based on $a^x_{v_{out}}$ in $\calF_s(v,w)$, while keeping $a^x_{w_{out}} = a_w$. Note that all assets $a^x_{v_{out}} \ge L_{v_{out}} = L_v$ yield the same clearing state and, thus, the same set of banks in default. We denote the set of solvent banks for $a^x_{v_{out}} = L_v$ by $S_0$. At $a^x_{v_{out}}$ the set of solvent banks changes, since $v_{out}$ becomes insolvent. As we start to decrease $a^x_{v_{out}}$ in the interval $[a_v,L_v]$, we monotonically decrease the available assets in $\calF_s(v,w)$. Thus, the set of solvent banks will shrink. Breakpoints at which additional banks become insolvent generate what we call the \emph{default hierarchy} of $\calF_s(v,w)$. 
\begin{definition}
    A \emph{default hierarchy} of $\calF_s(v,w)$ is given by 
    \begin{itemize}
        \item breakpoints $a_v^{(0)},\ldots,a_v^{(l)}$ with $a_v^{(0)} = \infty > L_v = a_v^{(1)} \ge a_v^{(2)} \ge \ldots \ge a_v^{(l-1)} \ge a_v^{(l)} = a_v$, and 
        \item subsets $V \supseteq S^{(0)} \supsetneq S^{(1)} \supsetneq \ldots \supsetneq S^{(l)}$ 
    \end{itemize}
    such that $S^{(j)}$ is the set of solvent banks in $\calF_s(v,w)$ if and only if $a^x_{v_{out}} \in [a_v^{(j+1)}, a_v^{(j)})$ and $a^x_{w_{out}} = a_w$, for each $j \in [l]$.
\end{definition}
The solvency intervals $[a_v^{(j+1)},a_v^{(j)})$ are closed (only) on the lower end -- solvency means $a_b \in [L_b,\infty)$, which is closed on the lower end for each $b \in V$. Since each $S^{(j)}$ is a strict subset of each $S^{(j')}$ with $j' < j$, it follows that $l \le n$. We show the following result.

\newcommand{\ot}{\leftarrow}

\begin{algorithm}[t]
\DontPrintSemicolon
\caption{Computing $S^{(j)}$} \label{alg:defaultBanks} 
\SetKwInOut{Input}{Input}\SetKwInOut{Output}{Output}

\Input{Split network $\calF_s(v,w)$, bank $v$, breakpoint $a_v^{(j)}$, solvent banks $S^{(j-1)}$}
$D^{(0)} \ot V_s \setminus S^{(j-1)}$, $r \ot 0$\;
\Repeat{$D^{(r)} = D^{(r-1)}$}{
    $\vecp^{(r)} \ot$ clearing state in $\calF_s(v,w)$ with $a^x_{v_{out}} = a_v^{(j)}$ and banks $D^{(r)}$ are in default\;
    $\veca^{(r)} \ot$ vector of total assets resulting from $\vecp^{(r)}$\;
    $T_0 \ot \{v\}, X_0 \ot \emptyset, I_0 \ot \{ b \in V \mid a^{(r)}_b < L_b\}, k \ot 0$ \;
    \Repeat{$T_k = \emptyset$}{
        $k \ot k+1$\;
        $I_k \gets I_{k-1} \cup \{t \in T_{k-1} \mid a_t^{(r)} = L_t\}$ and $X_k \gets X_{k-1} \cup T_{k-1}$\;
        $T_k \gets \{w' \in V \setminus X_k \mid a_{w'}^{(r)} \leq L_{w'}$ and $\exists$ $b \in T_{k-1}$ with $(b,w')\in E^+(b)\}$ 
    }
    $D^{(r+1)} \ot D^{(r)} \cup I_k$\;
    $r \ot r+1$
}
\Output{$S^{(j)} = V_s \setminus D^{(r)}$}
\end{algorithm}

\begin{proposition}
    \label{prop:default}
    The default hierarchy of $\calF_s(v,w)$ can be computed in polynomial time.
\end{proposition}
\begin{proof}
We proceed iteratively. $a_v^{(0)} = \infty$ is given, and $S^{(0)}$ can be found by computing the clearing state in $\calF_s(v,w)$ with $a_{v_{out}}^x = L_v$. Suppose we have computed $a_v^{(0)},\ldots,a_v^{(j)}$ and $S^{(0)},\ldots,S^{(j-1)}$. We show how to determine the next breakpoint $a_v^{(j+1)} \le a_v^{(j)}$ as well as the set $S^{(j)} \subsetneq S^{(j-1)}$ in polynomial time. Since the hierarchy has $l+1 \le n+1$ breakpoints and subsets in total, this gives an efficient algorithm.

\paragraph{Computing $S^{(j)}$.} 
We denote by $V_s$ and $E_s$ the node and edge sets of $\calF_s(v,w)$, respectively. To identify the banks that become insolvent when $a^x_{v_{out}}$ drops below $a_v^{(j)}$, we first consider all $b \in V_s$ in $\calF_s(v,w)$ such that the clearing state resulting from $a^x_{v_{out}} = a_v^{(j)}$ puts them precisely at their solvency frontier. Let $S_f = \{ b \in V_s \mid a_b = L_b\}$ be the set of these frontier banks. Suppose $a^x_{v_{out}}$ drops below $a_v^{(j)}$. Then $b \in S_f$ becomes insolvent immediately if every change of $a_{v_{out}}^x$ propagates to $b$. In contrast, $b \in S_f$ is safe against small changes of $a_{v_{out}}^x$ if all paths from $v_{out}$ to $b$ involve a (non-frontier) solvent bank from $S^{(j)} \setminus S_f$. Let $I \subseteq S_f$ be the \emph{critical banks}, i.e., the banks $b$ such that there is at least one path from $v_{out}$ to $b$ with no bank from $S_i \setminus S_f$. \autoref{alg:defaultBanks} uses a breadth-first-search (BFS) in the inner repeat-loop to identify these critical banks. 
In addition to this routine, we use a modified version of the clearing state algorithm for networks with default cost from~\cite{RogersV13}. In our modification we assume that all critical banks (that are precisely at the solvency frontier) become insolvent and, thus, a default cost reduction is applied to their assets. Such a reduction further decreases the overall assets in $\calF_s(v,w)$ and, thus, might trigger further banks become either insolvent or reach the solvency frontier and join $S_f$. We use the outer repeat-loop in \autoref{alg:defaultBanks} to repeatedly enlarge the set of banks $D^{(r)}$ that are affected by insolvency (either directly or by being a critical bank in the current iteration) and apply a default cost reduction to their assets. Thereby, the algorithm correctly identifies all banks that slide into default once $a^x_{v_{out}}$ drops below $a_v^{(j)}$.

Since $D^{(r)}$ is growing monotonically, the outer repeat-loop requires at most $n$ iterations. The modified algorithm for the clearing state $\vecp^{(r)}$ requires to solve at most $n$ linear programs of polynomial size. The BFS in the inner repeat-loop can be implemented in time linear in the size of $\calF_s(v,w)$. Overall, the algorithm runs in polynomial time.

\paragraph{Computing $a_v^{(j+1)}$.} Given the set $S^{(j)}$, we search for the smallest value of $a_{v_{out}}^x$ such that we keep the banks in $S^{(j)}$ solvent. 
It can be found by solving LP~\eqref{eq:LPforLB} below.
\begin{equation}
    \label{eq:LPforLB}
    \begin{aligned}
        \text{Min.\ } \; &\displaystyle a^x_{v_{out}} &\\
        \text{s.t.\ } \; & a_b = a_{b}^x + \sum_{(b',b) \in E_s^-(b)}r_{b'} \cdot \ell_{(b',b)} &\forall b\in V_s\\
        & a_b \ge L_b & \forall b\in S^{(j)}\\
        & r_b = 1 & \forall b \in S^{(j)}\\
        & a_b \le L_b & \forall b\in V_s \setminus S^{(j)}\\
        & r_{b} = \delta \cdot a_b/L_{b} &\forall b\in V_s \setminus S^{(j)} \\
        & a_b \ge 0 & \forall b\in V_s
    \end{aligned}
\end{equation}

The first set of constraints defines the assets for each node as the given external assets and the incoming payments, defined by the recovery rate of the debtor node and the liability. For solvent banks, we require the solvency condition and, as a consequence, a recovery rate of 1. For insolvent banks, we assume the assets to be lower than the liabilities and a recovery rate that suffers from a default cost reduction. 

While the LP includes all fixed-point constraints for the clearing state, our objective is to \emph{minimize} $a_{v_{out}}^x$. Hence, an optimal solution of the LP represents payments, assets, and recovery rates that correspond to some fixed point, but not necessarily the clearing state which is the \emph{maximal} one. 

Nevertheless, we show that the optimal solution correctly identifies the minimal value for $a_{v_{out}}^x$. Clearly, by setting the true value $a^x_{v_{out}} = a_v^{(j+1)}$ along with the resulting clearing state, we obtain a feasible solution for LP~\eqref{eq:LPforLB}. As such, the optimal solution of the LP can only be lower than $a_v^{(j+1)}$. Suppose for contradiction that the optimal solution has value $q < a_v^{(j+1)}$. Let us replace the values of $a_b$ and $r_b$ by the ones in the clearing state resulting from $a^x_{v_{out}} = q$. This pointwise increases all recovery rates and assets, since the clearing state is the supremum of the lattice of fixed points. If it results in a feasible solution, then $a_v^{(j+1)} \le q$, a contradiction. Otherwise, we must violate a constraint $a_b \le L_b$ for some $v \in S^{(j)}$, since the clearing state increases recovery rates and assets. However, we argued above that the clearing state for $a_v^{(j+1)}$ satisfies all these constraints. When the external assets of $v_{out}$ are decreased to $q < a_v^{(j+1)}$, it is easy to see that the clearing state is non-increasing in all entries. Thus, $a_b \le L_b$ must hold for all $b \in S_{i+1}$ also in the clearing state for $q$, a contradiction. Hence, although the optimal LP solution might not be a clearing state, it correctly yields the value $a_v^{(j+1)}$.
\end{proof}

\subsection{Computing a Single Trade}
\label{sec:trade}
In this section, we show how to compute a trade 
that gives both $v$ and $w$ (and, hence, each bank in the network) at least the post-trade assets of any creditor-positive trade. Our main result is the following theorem. For the proof we construct appropriate LPs for each part of the default hierarchy to check for existence and optimize assets of a trade within that part.
\begin{theorem}
    \label{thm:singleTradeVar}
    Suppose we are given a financial network $\calF$ with default cost $\delta$, a trade edge $e = (u,v) \in E$, and a buyer $w \in V$, for which a creditor-positive trade exists. There is a polynomial-time algorithm to compute a trade that achieves at least the post-trade assets for all banks of any creditor-positive trade.
\end{theorem}
\begin{proof}
    The task here is to compute the a fraction $\beta$ of edge $e$ that should be traded to $w$ (if any) and a haircut rate $\alpha$ to obtain good post-trade assets $a'_v$. Our approach uses the default hierarchy and solves an LP for each interval $[a_v^{(j+1)}, a_v^{(j)})$ and set $S^{(j)}$ in the hierarchy. As the objective, we always strive to maximize the post-trade assets $a'_v$.

    For the banks $V \setminus \{v,w\}$ we only include fixed-point and default cost constraints. Since $S^{(j)}$ is a set of solvent banks from $\calF_s(v,w)$, it contains banks $v_{in}, w_{in}$ as well as, potentially, $v_{out}$ and $w_{out}$, but neither $v$ nor $w$. We disregard these banks from consideration here, focus on $S^{(j)} \cap V$, and deal with $v$ and $w$ separately below. Given the set $S^{(j)}$ of solvent banks, all our LPs feature the following basic fixed-point constraints for the post-trade assets $\veca'$ and recovery rates $\vecr'$.
    
    \begin{equation}
        \label{eq:basicFP}
        \begin{aligned}
            a_{b}' &= a_{b}^{x\prime} + \sum_{(b',b) \in E^-(b)} r'_b \cdot \ell_{(b',b)} &\forall b\in V \setminus\{v,w\} \\
            a'_b &\ge L_b & \forall b \in S^{(j)} \cap V \\
            r'_b &\leq a'_b/L_b & \forall b \in V\\ 
            r'_b &= 1 & \forall b \in S^{(j)} \cap V\\ 
            r'_b &\leq \delta \cdot a'_b/L_b & \forall b \in V \setminus (S^{(j)} \cup \{v,w\})\\
            a_b^{x\prime} &= a_b^x & \forall b\in V \setminus \{v,w\} \\
            r_b' &\in [0,1] & \forall b\in V &
        \end{aligned}
    \end{equation}
   
    Moreover, we require that the post-trade assets satisfy
    \begin{equation}
        \label{eq:moreThanCP}
        a_w' \geq a_w \hspace{1cm} \text{ and } \hspace{1cm}  a_v' \in [a_v^{(j+1)},a_v^{(j)}].
    \end{equation}
    The first constraint requires a Pareto-improvement for the assets of $w$. It also implies that the trade we compute might not necessarily be \emph{creditor-positive} itself -- it can also be a positive one. In the default hierarchy we have enumerated all possible scenarios for any creditor-positive trade in terms of solvent banks and assets for $v$. As such, we set up and solve an LP for each of these scenarios $j \in [l]$ with the objective to maximize $a'_v$, and keep the best solution. In this way, the resulting trade will deliver the (weak) Pareto-improvement in post-trade assets for $v$ and $w$ (and, hence, for all banks in the network) over all creditor-positive trades.

    The second constraint uses the closed interval. Consequently, we might find a trade with $a'_v = a_v^{(j)}$. In this case, we consider the same trade with solvent banks $S^{(j-1)}$ instead of $S^{(j)}$. Then a strict superset of banks are solvent, which can only create (weakly) better assets for $v$ and $w$ (and, hence, all banks in the network). As such, the trade is Pareto-dominated by the solution of our LP for part $j-1$. 

    Now consider the following set of constraints for the trade variable $\beta$, haircut rate $\alpha$, the result on the external assets, and fixed-point conditions for $v$ and $w$:
    \begin{equation}
    \label{eq:directLP}
            \begin{aligned}
            \alpha, \beta &\in [0,1]\\
            \alpha \beta \ell_e &\le a_w^x\\
            a_w^{x\prime} &= a_w^x - \alpha \beta \ell_e \\
            a_v^{x\prime} &= a_v^x + \alpha \beta \ell_e \\
            r_{v}' &= \delta_v \cdot a_v'/L_{v}\\
            r_{w}' &= \delta_w \cdot a_{w}'/L_{w} \\
            \end{aligned} \hspace{2cm} 
            \begin{aligned}
            a_{w}' &= a_{w}^{x\prime} + \left(\sum_{(b,w) \in E^-(w)} r'_b \cdot \ell_{(b,w)}\right) + r_u'\beta\ell_e \\
            a_{v}' &= a_{v}^{x\prime} + \left(\sum_{(b,v) \in E^-(v) \setminus \{e\}} r'_b \cdot \ell_{(b,v)}\right) + r_u'\ell_e - r_u'\beta\ell_e \\
        \end{aligned}
    \end{equation}
    where $\delta_v = 1$ for $i = 0$ when $a'_v \in [L_v,\infty)$ and $\delta_v = \delta$ otherwise. Similarly, $\delta_w = 1$ if $a_w \geq L_w$, and $\delta_w = \delta$ otherwise.

    The constraints ensure $\alpha$ and $\beta$ are feasible fractions, that the return can be paid from the external assets of $w$, and define the new external assets after the return payment, the new total assets after the fractional trade of incoming edge $e$, as well as the resulting recovery rates depending on solvency of $v$ and $w$.

    The challenge is that the constraints for the assets of $a'_v$ and $a'_w$ in~\eqref{eq:directLP} are not linear. Based on our characterization in Proposition~\ref{prop:structure}, we consider two cases with $\beta = 1$ and $\alpha\beta\ell_e = \rho = a_w^x$. First, suppose $\beta = 1$, then the constraints~\eqref{eq:directLP} become linear. Thus, we solve the LP to maximize $a'_v$ subject to constraints in~\eqref{eq:basicFP}, \eqref{eq:moreThanCP}, \eqref{eq:directLP}. Second, if the return $\rho = \alpha\beta\ell_e = a_w^x$, then the external assets become $a_w^{x\prime} = 0$ and $a_v^{x\prime} = a_v^x + a_w^x$. The constraints for $a'_v$ and $a'_w$ are not linear. We consider two subcases here. First, suppose $r'_u = 0$, then we get linear constraints and solve the resulting LP to maximize $a'_v$ subject to constraints in~\eqref{eq:basicFP},\eqref{eq:moreThanCP},\eqref{eq:directLP}.
    
    Second, if $r'_u > 0$, we apply a substitution: Let $y = r'_u \beta \ell_e$, then $\beta = y/(r_u' \ell_e)$ and, hence, $\alpha = r'_u a_w^x$. By combining the equations for external and total assets of $w$ and $v$, respectively, we transform~\eqref{eq:directLP} into equivalent linear constraints
    \begin{equation}
        \label{eq:varLP}
        \begin{aligned}
            y &\in [0, r'_u \ell_e]\\
            a_{w}' &= \left(\sum_{(b,w) \in E^-(w)} r'_b \cdot \ell_{(b,w)}\right) + y \\
            a_{v}' &= a_v^x + a_w^x + \left(\sum_{(b,v) \in E^-(v) \setminus \{e\}} r'_b \cdot \ell_{(b,v)}\right) + r_u' \ell_e - y\\
            r_{v}' &= \delta_v \cdot a_v'/L_{v}\\
            r_{w}' &= \delta_w \cdot a_{w}'/L_{w} \\
        \end{aligned}
    \end{equation}
    As such, we solve the resulting LP to maximize $a'_v$ subject to constraints in~\eqref{eq:basicFP},\eqref{eq:moreThanCP},\eqref{eq:varLP}.
    
    Note that since $a'_w \ge a_w$, an optimal solution might have $a'_w \ge L_w$. In fact, such a solution might only evolve if we set $\delta_w = 1$ in the first place. As such, if $a_w < L_w$, we solve the LPs in both cases twice -- once with $\delta_w = \delta$ and once with $\delta_w = 1$. In the latter case, we check if the optimal solution satisfies $a'_w \ge L_w$. If so, we obviously prefer the latter solution (since it Pareto-dominates the former in terms of assets), otherwise we keep the former.

    We do not include an explicit optimization to obtain payments, assets, and recovery rates from the \emph{clearing state}. Rather, we only include fixed-point constraints. Let us argue that this is unproblematic. The only true upper-bound constraint in our LPs will be $a'_v \le a_v^{(j)}$. Now suppose we compute an optimal LP solution with $a'_v < a_v^{(j)}$. If it does not represent the clearing state (i.e., the maximal fixed point), we can replace all payments, assets and recovery rates by the ones in the clearing state. This does not violate any fixed-point constraints and can only increase $a'_v$ and $a'_w$. Thus, w.l.o.g.\ we can assume an optimal LP solution that uses clearing state payments. Conversely, consider an optimal solution with $a'_v = a_v^{(j)}$. A clearing state might violate the upper-bound constraint. However, as argued above, in this case we rather resort to part $i-1$ in the hierarchy which can only improve the conditions for all banks.

    It is possible that the LPs in both cases are infeasible, which proves that in the part $j \in [l]$ of the default hierarchy no creditor-positive trade can exist.
\end{proof}
By Proposition~\ref{prop:noPositive}, networks without default cost allow no positive trades. Then the trade computed by the algorithm must itself be a creditor-positive one. As such, we can compute an \emph{optimal} creditor-positive trade in the network, or decide that no such trade exists.
\begin{corollary}
    \label{cor:singleTradeVar}
    Given a financial network $\calF$ without default cost, a trade edge $e = (u,v) \in E$, and a buyer $w \in V$, there is a polynomial-time algorithm to compute an optimal creditor-positive trade or decide that none exists.
\end{corollary}

\section{Multi-Trades of Incoming Edges}
\label{sec:incoming}
In this section, we generalize the results from the previous section to multi-trades of incoming edges. Recall that for a multi-trade of incoming edges, we are given creditor $v$ and buyer $w$. We denote $v$'s incoming edges by $e_1,e_2,\ldots,e_k$ and the corresponding debtors by $u_1,u_2,\ldots,u_k$. Our goal, again, is to find a trade that is at least as good (weakly Pareto-improving for all banks) as the optimal creditor-positive trade. Our main result is the following theorem.
\begin{theorem}
    \label{thm:incoming}
    Suppose we are given a financial network $\calF$ with default cost $\delta$, a creditor $v$, a buyer $w \in V$, for which a creditor-positive multi-trade of incoming edges exists. There is a polynomial-time algorithm to compute a trade that achieves at least the post-trade assets for all banks of any creditor-positive multi-trade of incoming edges.
\end{theorem}
For networks without default cost, this again allows to compute the optimal creditor-positive trade or decide that none exists.
\begin{corollary}
    \label{cor:singleTradeFix}
    Given a financial network $\calF$ without default cost, a creditor $v \in V$, and a buyer $w \in V$, there is a polynomial-time algorithm to compute an optimal creditor-positive multi-trade of incoming edges or decide that none exists.
\end{corollary}
We show the result using a careful enumeration of parameters and repeated application of Theorem~\ref{thm:singleTradeVar}. More formally, given any multi-trade with incoming edges, we can assume w.l.o.g.\ $\rho = \sum_{i \in [k]} \alpha_i \beta_i \ell_{e_i} = \alpha \cdot \sum_{i \in [k]} \beta_i \ell_{e_i}$, i.e., $\alpha_i = \alpha$ for a suitably chosen $\alpha \in [0,1]$. We use this insight to show that we can restrict attention to trades that have a greedy structure. They can be interpreted as solving a fractional knapsack problem (c.f.~\cite{HoeferVW24} for a similar result for binary trades). 
\begin{lemma}
    \label{lem:greedyVar}
    For every multi-trade of incoming edges, consider the edges $i \in [k]$ sorted in order of non-decreasing recovery rate $r'_{u_i}$. There is an equivalent trade with a number $i' \in [k]$ and traded fractions $\hat{\vecBeta}$ such that $\beta_i = 1$ for $i < i'$ and $\beta_i = 0$ for $i > i'$.
\end{lemma}
\begin{proof}
Consider any multi-trade of incoming edges with the vector of traded fractions $\vecBeta$ and a haircut rate of $\alpha$. The incoming edges of $v$ generate a total incoming payment of $\tau = \sum_{i \in [k]} r'_{u_i} \ell_{e_i}$. From this total payment, $w$ receives $\gamma = \sum_{i \in [k]} r'_{u_i} \beta_i \ell_{e_i}$, and $v$ receives $\tau - \gamma$. Also, $v$ gets $\rho$ as return payment from $w$ such that $\rho = \alpha \sum_{i \in [k]} \beta_i\ell_{e_i} \le a_w^x$. 

Now construct an adjusted trade with $\hat{\vecBeta}$ and $\hat{\alpha}$ that satisfies the properties of the lemma, i.e., with $i' \in [k]$ such that edges $i < i'$ are traded fully with $\beta_i' = 1$, and edges $i > i'$ not at all with $\beta'_i = 0$. Edge $i'$ is the only one that is potentially traded fractionally, i.e., $\hat{\beta}_{i'} \in [0,1]$. We choose $i'$ and $\hat{\beta}_{i'}$ such that 
\[
        \sum_{i \in [k]} r_{u_i}' \hat{\beta}_i \ell_{e_i} = \sum_{i < i'} r_{u_i}' \ell_{e_i} + r_{u_{i'}} \hat{\beta}_{i'} \ell_{e_{i'}} = \gamma,
\]
so the adjusted trade has the same incoming payments for $v$ and $w$ from the debtors $u_i$. Moreover, $\sum_{i \in [k]} \hat{\beta}_i \ell_{e_i} \ge \sum_{i \in [k]} \beta_i \ell_{e_i}$, since the adjusted trade focuses on edges with smallest recovery rates. Hence, it is required to trade more total liability to obtain the same incoming payments for $w$ from the $u_i$. Therefore, we can choose $\alpha \in [0,1]$ such that 
\[
    \hat{\alpha} \cdot \sum_{i \in [k]} \hat{\beta}_i \ell_{e_i} = \rho, \hspace{1cm} \text{where } \hat{\alpha} \le \alpha.
\]
The adjusted trade is equivalent to the original one, i.e., has the same total payments from the $u_i$ to $v$ and $w$ and the same return. It trades edges with smallest recovery rates.
\end{proof}

Using this lemma, we present a proof of our main result.
\begin{proof}[Proof of Theorem~\ref{thm:incoming}]
Consider the default hierarchy. In the split network $\calF_s(v,w)$, a move of incoming edges from $v$ to $w$ changes only the assets of $v_{in}$ and $w_{in}$ but not the clearing state or asset levels of other banks. Hence, an optimal creditor-positive multi-trade of incoming edges has an equivalent representation with the same assets and payments in $\calF_s(v,w)$.

An algorithmic consequence of Lemma~\ref{lem:greedyVar} is the following. Suppose there is a \emph{structure oracle} that returns for an optimal creditor-positive trade, (1) the relevant part $j$ of the default hierarchy for the post-trade assets of $v$, and (2) the ordering of edges $e_i$ w.r.t.\ recovery rates $r'_{u_i}$. Then we can enumerate all $k$ possibilities for the index $i'$. For each $i' \in [k]$, we assign all edges $e_i$ with $i < i'$ fully to $w$ and use the algorithm from the previous section to compute a single trade for edge $e_{i'}$. More formally, for each $i' \in [k]$, we set up the LPs for part $j$ of the default hierarchy to compute a single trade of edge $e_{i'}$ as in Theorem~\ref{thm:singleTradeVar} above, with the adjustment that all edges $e_i$ with $i < i'$ are assigned fully to $w$. Unless $i' = k$, Proposition~\ref{prop:structure} shows that we only need to check the case that $\rho = a_w^x$. For the correct choice of the fractional edge $e_{i'}$, the best of the LP will yield a multi-trade of incoming edges that Pareto-dominates the optimal creditor-positive one.

In the remainder of the proof, we show how to enumerate a polynomial number of combinations that include the correct output of our structure oracle. Thus, by applying the approach from the previous paragraph to all combinations, we obtain an efficient algorithm and prove the theorem.

Consider $\calF_s(v,w)$ and each part of the default hierarchy individually. Suppose we vary the assets $a'_v = a^x_{v_{out}}$ within an interval $a'_v \in [a_v^{(j+1)},a_v^{(j)})$. Then there are no (in-)solvency events. As such, each bank $b \in V \setminus (S^{(j)} \cup \{v,w\})$ is insolvent and has outgoing assets $\delta a_b$, i.e., linear in the total assets of $b$, whereas $b \in S^{(j)}$ is solvent and acts as a sink for the incoming assets with outgoing assets $L_b$. Instead of a common sink node $v_{in}$ let us consider an individual sink node $v_i$ for each edge $e_i$. It is a well-known property of the Eisenberg-Noe model (see, e.g.~\cite{FroeseHW23,PappW21EC}) that -- as long as no bank becomes solvent -- increasing the external assets of a source node leads to a linear change in total assets at every sink node. It is easy to see that this property continues to hold for networks with default cost. As such, the assets of each $v_i$ increase linearly in $a_v'$ within part $j$ of the hierarchy. Since $a_{v_i} = r_{u_i} \ell_{e_i}$, the recovery rate $r_{u_i}$ also grows linearly in $a_v'$.

As such, all recovery rates are linear functions $r_{u_i} : [a_v^{(j+1)},a_v^{(j)}) \to [0,1]$. Slopes and offsets of these functions can be found easily. For a given ordering of $r_{u_i}(a'_v)$ to become infeasible, there must be a true intersection of (at least) two of the linear functions. For $k$ linear functions there are at most $k^2$ true intersections. These intersections and the resulting orderings are computed by solving linear equations.

Overall, there are $l+1$ parts in the hierarchy. Each part gives rise to at most $k^2$ different orderings of the recovery rates. The combinations can be enumerated in polynomial time. 
\end{proof}

\section{Multi-Trades of Outgoing Edges}
\label{sec:outgoing}
Recall that for a multi-trade of outgoing edges, we are given a financial network, a debtor $u$ and buyer $w$, and we denote $u$'s creditors by $v_1,v_2,\ldots,v_k$. The goal is to compute a Pareto-positive (fractional) multi-trade of outgoing edges of $u$, i.e., none of the buyer $w$ or the creditors $v_i$ are harmed while at least one $v_i$ strictly profits. We strive to find a Pareto-positive trade that maximizes the total assets (or, equivalently, the total increase in assets) of all creditors, i.e., $\max \sum_{i \in [k]} a'_{v_i}$. 

For networks with the default cost $\delta \in [0,1)$ we show the following hardness result. It extends to the case when we restrict attention to trades with a natural form of returns. We say $w$ \emph{pays excess returns} if it gives a return $\rho_i \ge r'_u \cdot \beta_i \ell_{e_i}$ to each bank $v_i$. An excess return $\rho_i$ exceeds the payments $r'_u \cdot \beta_i \ell_{e_i}$ received by $w$ on the traded fraction of the edge in the clearing state of the post-trade network.
\begin{theorem}
    \label{thm:NPC}
    It is \classNP-hard to compute a Pareto-positive
    \begin{enumerate}
        \item multi-trade of outgoing edges that maximizes the total assets of all creditors,
        \item multi-donation that maximizes the total assets of all banks,
    \end{enumerate}
    The hardness results apply also when we restrict attention to trades with excess returns.
\end{theorem}
\begin{proof}
    We start with the proof for multi-trades and show a reduction from Set Packing (c.f., ~\cite{HoeferVW24}). In an instance of Set Packing, there are $m$ elements $x_1,\ldots,x_m$ and $k$ sets $S_1,\ldots,S_k$. Each set contains a subset of the elements. Moreover, there is an integer $k > 0$. We need to decide whether or not one can select $l$ of the sets such that all selected sets remain mutually disjoint. 
    
    \begin{figure}
    \begin{center}
    \resizebox{0.96\textwidth}{!}{
    \begin{tikzpicture}[scale=0.2, >=stealth', shorten >=1pt, auto, node distance=2cm]
        \node[circle, draw, inner sep=3.3] (v) at (0,0) {$u$};
        \node[bank] (s1) [below left=1cm and 3.3cm of v] {$S_1$};
        \node[bank] (s2) [right=2cm of s1] {$S_2$};
        \node[bank] (s3) [right=of s2] {$S_3$};
        \node (sdots) [right=0.4cm of s3] {$\cdots$};
        \node[bank] (sl) [right=2cm of s3] {$S_k$};
        \node[circle, draw, inner sep=2] (u3) [below=2.42cm of v] {$x_3$};
        \node[circle, draw, inner sep=2] (u2) [left=of u3] {$x_2$};
        \node[circle, draw, inner sep=2] (u1) [left=of u2] {$x_1$};
        \node[circle, draw, inner sep=2] (u4) [right=of u3] {$x_4$};
        \node (edots) [right=0.4cm of u4] {$\cdots$};
        \node[circle, draw, inner sep=1.8] (um) [right=of u4] {$x_m$};
        \node[circle, draw, inner sep=2.8] (w) [below=1cm of u3] {$w$};
        \node[rectangle, draw=black, inner sep=2.5pt] (xw) [below=0.1cm of w] {\footnotesize $lM$};
        \draw[->] (v) --node[left=0.3cm, pos=0.4] {\footnotesize$M$} (s1);
        \draw[->] (v) --node[left=0.05cm, pos=0.7] {\footnotesize$M$} (s2);
        \draw[->] (v) --node[right=0.05cm, pos=0.7] {\footnotesize$M$} (s3);
        \draw[->] (v) --node[right=0.3cm, pos=0.4] {\footnotesize$M$} (sl);
        \draw[->] (s1) --node[left, pos=0.3] {\footnotesize{$1$}} (u1);
        \draw[->] (s1) --node[left=0.1cm, pos=0.3] {\footnotesize{$1$}} (u2);
        \draw[->] (s1) --node[right=0.3cm, pos=0.1] {\footnotesize{$1$}} (u3);
        \draw[->] (s2) --node[right=0.2cm, pos=0.25] {\footnotesize{$1$}} (u1);
        \draw[->] (s2) --node[right, pos=0.4] {\footnotesize{$1$}} (u3);
        \draw[->] (s2) --node[right=0.2cm, pos=0.05] {\footnotesize{$1$}} (u4);
        \draw[->] (s3) --node[left, pos=0.1] {\footnotesize{$1$}} (u3);
        \draw[->] (s3) --node[left, pos=0.4] {\footnotesize{$1$}} (u4);
        \draw[->] (sl) --node[left, pos=0.4] {\footnotesize{$1$}} (um);
        \draw[->] (s1) to[in=200, out=200, looseness=2.5] node[left, midway] {\footnotesize{$M-|S_1|$}} (w);
        \draw[->] (sl) to[in=-20, out=-20, looseness=2.5] node[right, midway] {\footnotesize{$M-|S_k|$}} (w);
        \draw[->] (u1) --node[below=0.3cm,pos=0.3] {\footnotesize $1$}(w);
        \draw[->] (u2) --node[right=0.1cm,pos=0.1] {\footnotesize $1$}(w);
        \draw[->] (u3) --node[left=-0.1cm,pos=0.4] {\footnotesize $1$} (w);
        \draw[->] (u4) --node[left=0.1cm,pos=0.2] {\footnotesize $1$}(w);
        \draw[->] (um) --node[below=0.3cm,pos=0.4] {\footnotesize $1$}(w);
    \end{tikzpicture}
    }
    \end{center}
    \vspace*{-0.5cm}
    \caption{Every bank $S_i$ has an edge to $w$ with liabilities $M-|S_i|$. Default cost is given by $\delta < 1$.}
    \label{fig:default-NP}
    \end{figure}
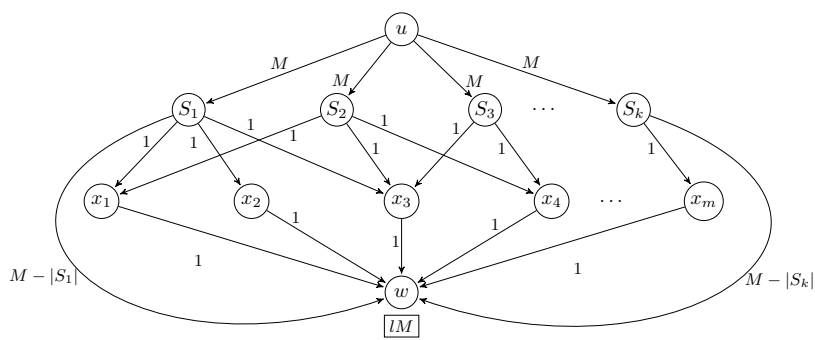

    Given an instance of Set Packing, we construct a network as depicted in Figure~\ref{fig:default-NP}. There is default cost in the sense that $\delta < 1$. We claim that there exists a Pareto-positive multi-trade of outgoing edges with sum of profits $lM$ if and only if there exists a set packing of size $l$.

    First assume there exists a set packing $P$ of size $l$. Then, fully trade the edge $e_i = (u,S_i)$, for each $S_i \in P$, and set the return to $M$. Clearly, all $S_i$ included in the trade are solvent and thus able to make payments. Since $P$ is a feasible set packing, all banks $S_i$ included in the trade have pairwise disjoint creditors $x_j$. For these reasons, the total return is paid back to $w$ and the sum of profits is $lM$.

    For the other direction, assume there exists a multi-trade with sum of profits $lM$. Since $w$ must not be harmed and no cycles can emerge, the return paid by $w$ must flow back to $w$ entirely. Hence, every bank $S_i$ included in the trade with strictly positive return must receive $\rho_i=M$. Otherwise, $S_i$ would be insolvent, and part of the return would be lost due to default cost. Since all banks $S_i$ with $\rho_i=M$ forward all return back to $w$, their creditors $x_j$ must be pairwise distinct. Together with the fact there are $k$ such banks, this implies that there exists a set packing of size $l$.

    To adjust the construction for multi-donations, remove node $u$ and all edges $(u,S_i)$. Since $w$ wants to maximize the total assets of all banks, intuitively, donating to banks $S_i$ is more attractive than donating to banks $x_j$. Every unit donated to a bank $S_i$ (optimally) can benefit two banks ($S_i$ and one $x_j$, rather than only $x_j$) before returning to $w$. For a Pareto-positive donation, we must ensure that all donations fully return to $w$ eventually. Hence, all banks except $w$ must forward all their assets, and no money must be lost due to default cost. These conditions allow to apply the same arguments as above. In an optimal donation, $w$ only donates payments of $M$ to a subset of the banks $S_i$ and achieves total assets of $3lM$ in the network if and only if the instance of Set Packing is solvable.

%
    Finally, since in our instances $r'_u = 0$ always, all returns given by $w$ are trivially excess returns.
\end{proof}

For a single trade in a network without default cost (see, e.g., the example in Fig.~\ref{fig:ex-intro}), excess returns are necessary for a creditor-positive trade (since otherwise $w$ would strictly profit, and hence $v$ does not, by Proposition~\ref{prop:noPositive}). For multi-trades of outgoing edges, limiting attention to excess returns can represent a true restriction. Consider a network with edges $(u,v_1)$, $(u,v_2)$, $(v_1,v_2)$, and $(v_2,w)$, all with $\ell_e = 4$. $u$ has external assets of 2, $w$ has external assets of 4. The total assets become $a_u = 2$, $a_{v_1} = 1$, $a_{v_2} = 2$ and $a_w = 6$. There is an optimal Pareto-positive trade where both $(u,v_1)$ and $(u,v_2)$ are traded entirely to $w$. The return is \emph{4 for $v_1$ and 0 for $v_2$}. This maximizes the total post-trade assets of $v_1$ and $v_2$, which become $a'_{v_1} = a'_{v_2} = 4$. However, these are \emph{not excess returns}: $w$ gives $v_2$ no return (directly) even though it receives a payment of 1 from $u$ on the traded edge $(u,v_2)$.

Even though excess returns can represent a restriction, they arguably constitute a natural subclass of trades. Moreover, they have favorable properties in terms of computational complexity. Suppose we restrict attention to Pareto-positive multi-trades of outgoing edges \emph{in which $w$ pays excess returns}. Then the problem of computing such trades becomes efficiently solvable in networks without default cost. In particular, this also holds true for multi-donations, since they can be seen as multi-trades with excess returns. As such, the hardness in Theorem~\ref{thm:NPC} crucially relies on default cost. 

\begin{theorem}
    \label{thm:excess}
    Given a financial network $\calF$ without default cost, debtor $u$, and buyer $w$, there is a polynomial-time algorithm to compute a Pareto-positive multi-trade of outgoing edges with excess returns that maximizes the total assets of all creditors, or decide that none exists.
\end{theorem}
\begin{proof}
    We define $\eta_i \ge 0$ as the excess return, i.e., the portion paid to $v_i$ \emph{in addition} to the post-trade payments on the traded edge fraction $r'_u \cdot \beta_i \ell_{e_i}$. Thus, the total return for $v_i$ is $\rho_i = \eta_i + r_u \beta_i \ell_{e_i} \le \beta_i \ell_{e_i}$, and hence $\eta_i \le (1-r'_u) \beta_i \ell_{e_i}$. We just postulate the constraint $0 \le \eta_i \le (1-r'_u) \ell_{e_i}$ and thereby implicitly define $\beta_i = \eta_i/((1-r'_u)\ell_{e_i})$. 

    Observe that we assume $r'_u < 1$, since otherwise $\beta_i$ is not well-defined. To justify this assumption, suppose we see $r'_u = 1$ after a Pareto-optimal trade with excess returns. Then all edges are fully paid for and $u$ is solvent. On each edge, we can return at most $\rho_i \le \beta_i \ell_{e_i}$, and due to excess returns, we must have $\rho_i \ge 1 \cdot \beta_i \ell_{e_i}$. As such, $w$ returns to $v$ precisely the payments received from $u$ via the traded portion. Thus, clearly, $r_u = 1$ must have been true \emph{before} the trade. In this case, it is easy to see that a trade of outgoing edges of $u$ cannot be Pareto-positive. The condition $r_u = 1$ is observable from the initial clearing state in $\calF$. Hence, we can indeed assume w.l.o.g.\ that $r'_u < 1$.

    Notably, we define $\beta_i$ in the way that the overall return paid to $v_i$ is \emph{exactly} $\beta_i \ell_{e_i}$. Since the excess return is $\eta_i = \beta_i (1-r'_u)\ell_{e_i} = \beta_i \ell_{e_i} - r'_u \beta_i \ell_{e_i}$, the overall return paid to $v_i$ is $\rho_i = \eta_i + r'_u \beta_i \ell_{e_i} = \beta_i \ell_{e_i}$. As such, we trade a portion of the edge whose liability exactly matches the return we pay. Moreover, 
    \begin{equation*}
    \sum_i \rho_i = \sum_i \eta_i + r_u \beta_i\ell_{e_i} = \sum_i \beta_i \ell_{e_i} = \sum_i \frac{\eta_i}{(1-r'_u) \ell_{e_i}} \ell_{e_i} \le a_w^x \quad \Longleftrightarrow \quad \sum_i \eta_i \le a_w^x (1-r'_u) 
    \end{equation*}
    Hence, we can express the upper bound on the total return as a linear constraint in the $\eta_i$ and $r'_u$. 

    For the constraints specifying the assets $a'_w$ and $a'_{v_i}$, we can assume that $v_i$ continues to keep the entire incoming edge $\ell_{e_i}$ and gets the additional return of $\eta_i$ as a donation from $w$. This is equivalent to $w$ taking a $\beta_i$-share of edge $\ell_{e_i}$ and then returning to $v_i$ the payment over that edge plus $\eta_i$. We only need to make sure that the donation of $\eta_i$ remains equivalent to a return payment that can be made using a partial assignment of edge $(u,v_i)$ to $w$ and is not too large (i.e., does not exceed an $(1-r'_u)$ share of the liability of the traded portion). Using this equivalent view, we can formulate the problem of finding an optimal Pareto-positive trade as the following linear program:
    
    \begin{equation}
        \label{eq:fracMultiTradeOutLP}
        \begin{aligned}
            \text{Max. } & \sum_{i\in[k]} a_{v_i}'\\
            &a_b' \geq a_b &\forall b\in V\\
            &a_b' = a_b^x + \sum_{(b',b)\in E^-(b)} r_{b'}'\cdot \ell_{(b',b)} &\forall b\in V \setminus \{v_1,\ldots,v_k,w\}\\
            &a_{v_i}' = a^x_{v_i} + \sum_{(b',v_i)\in E^-(v_i)} r_{b'}'\cdot \ell_{(b',v_i)} + \eta_i &\forall i \in [k]\\
            &a_w' = a_w^x + \sum_{(b',w)\in E^-(w)} r_{b'}'\cdot \ell_{(b',w)} - \sum_{i=1}^k \eta_i \\
            &\sum_{i=1}^k \eta_i \le a_w^x (1-r'_u) \\
            &\eta_i \in [0, (1-r'_u)\ell_{e_i}] & \forall i \in [k]\\
            &r_b' \leq a_b'/L_b &\forall b\in V\\
            &r_b' \in [0,1] &\forall b\in V
        \end{aligned}
    \end{equation}
%
    In particular, if the sum of assets is the same as in the clearing state of the pre-trade network $\calF$, we know that no Pareto-positive trade exists.
\end{proof}

For multi-donations, recall the representation as multi-trades of outgoing edges. We use an auxiliary source bank with no external assets and auxiliary edges to each bank $b \in V \setminus \{w\}$. As such, every donation represents an excess return. Thus, the construction in the previous theorem can be adjusted to yield the following corollary. In the LP, we only need to replace the objective function by the sum of all (post-trade) assets. Clearly, instead of \emph{all} banks, we can also optimize the total assets of any given subset of banks.

\begin{corollary}
    \label{cor:multiDonation}
    Given a financial network $\calF$ without default cost and a donor $w$, there is a polynomial-time algorithm to compute a Pareto-positive multi-donation that maximizes the total assets of any given subset of banks, or decide that none exists.
\end{corollary}

\section{Extensions and Open Problems}

We have designed algorithms to compute trades that achieve a Pareto-improvement over optimal creditor-positive trades in financial networks with default cost. Using our algorithms we can find optimal creditor-positive trades in networks without default cost, or show that they do not exist. Let us discuss several interesting consequences and direct extensions of our results.

\paragraph{Positive Trades of Incoming Edges.}
Creditor-positive trades, where a $a'_v > a_v$ for the creditor and $a'_w = a_w$ for the buyer, are a natural desiderata in networks without default cost and $\delta = 1$ (c.f.\ Proposition~\ref{prop:noPositive}). In networks with default cost and $\delta < 1$, it is easy to see that there are examples of \emph{positive} trades, in which \emph{both $v$ and $w$ improve strictly}. Towards this end, suppose we are given a larger target value $\omega > a_w$, and we strive to find trades that maximize $a'_v$ while ensuring $a'_w = \omega$. We call such a trade an optimal creditor-positive buyer-$\omega$-improved trade. It is straightforward to generalize Theorem~\ref{thm:incoming} to such trades. 
\begin{corollary}
    Suppose we are given a financial network $\calF$ with default cost $\delta$, a creditor $v$, a buyer $w \in V$, and an asset target $\omega > a_w$, for which a creditor-positive buyer-$\omega$-improved multi-trade of incoming edges exists. There is a polynomial-time algorithm to compute a trade that achieves at least the post-trade assets for all banks of any creditor-positive buyer-$\omega$-improved multi-trade of incoming edges.
\end{corollary}

Conceptually, our algorithm obtains \emph{some} trade that weakly Pareto-improves the assets over every trade in a class of ``benchmark trades'' -- in networks where at least one benchmark trade exists. The trade we compute is not necessarily guaranteed to be, say, the one that globally optimizes the sum of assets, or the assets of $v$ (but the assets of $v$ will be at least as high as after any benchmark trade).

It is an interesting open problem to directly optimize over positive multi-trades of incoming edges w.r.t.\ both the assets of $v$ and $w$, e.g., finding ones that maximize $a'_v + a'_w$. It is likely that this requires to develop new methods, even for single trades. In particular, the default hierarchy must be adjusted to represent an appropriate improvement of both agents. 

\paragraph{Unbounded Returns.}
In our paper, we assume that returns must be paid from external assets of $w$ to provide guaranteed liquidity to creditor(s) $v_i$. One might also imagine a scenario, in which the return is simply a regular claim with higher priority, which is only bounded by the liability of the traded claim (or, even more generally, has arbitrary liability) without connection to the external assets of $w$. In such scenarios, we can apply the adjustment outlined in Proposition~\ref{prop:structure} to obtain an equivalent trade in which all edges are traded completely with $\beta_i = 1$.
\begin{corollary}
    \label{cor:unboundedReturn}
    If the returns paid by $w$ are not bounded by $a_w^x$, then for every fractional multi-trade there is an equivalent trade for which all potential trade edges are traded to $w$ completely.    
\end{corollary}
It then remains to compute appropriate returns. For networks without default cost, we construct the following LP to compute optimal returns for a variety of objective functions.

\begin{proposition}
    Given a financial network $\calF$ without default cost and a buyer $w$, there is a polynomial-time algorithm to compute optimal trades with unbounded returns for a variety of objective functions.
\end{proposition}
\begin{proof}
    Suppose we trade a set $C \subseteq E$ of claims to bank $w$. Let $T = \{ b\in V \mid (u,b) \in C\}$ be the creditors of these claims. By Corollary~\ref{cor:unboundedReturn} we trade w.l.o.g.\ the entire set completely to $w$. Defining feasible returns then leads to the following set of linear fixed-point constraints.
    
    \begin{equation}
        \label{eq:unboundedReturns}
        \begin{aligned}
            &a_b' = a_b^x + \sum_{(b',b)\in E^-(b) \setminus C} r_{b'}'\cdot \ell_{(b',b)} + \rho_b &\forall b \in T\\
            &a_b' = a_b^x + \sum_{(b',b)\in E^-(b)} r_{b'}'\cdot \ell_{(b',b)} + \rho_b &\forall b \in V \setminus (T \cup \{w\})\\
            &a_w' = a_w^x + \sum_{(b',w)\in E^-(w)} r_{b'}'\cdot \ell_{(b',w)} + \sum_{(b',b) \in C} r_{b'}' \cdot \ell_{(b',b)} - \sum_{b \in T} \rho_b\\
            &r_b' \leq a_b'/L_b &\forall b\in V\\
            &r_b' \in [0,1] &\forall b\in V\\
            &0 \le \rho_b \le \sum_{(b',b) \in C} \ell_{(b',b)} &\forall b \in T
        \end{aligned}
    \end{equation}
    Here we still require that returns are bounded by the liability of the traded claims. This, however, can also be changed to (entirely) unbounded returns by dropping the upper bound in the last constraint. Moreover, we can add further natural (linear) constraints, such as Pareto improvement for $w$ ($a'_w \ge a_w$) or for all banks ($a'_b \ge a_b$ for all $b \in T$), or just an improvement in terms of total assets of $T$ ($\sum_{b \in T} a'_b \ge \sum_{b \in T} a_b$). Finally, any linear objective can be added, such as, e.g., maximizing total assets of the creditors $T$, or the sum of assets of $w$ and the creditors $T$, or a convex combination. The resulting LP can be solved efficiently.
\end{proof}
When returns are not bounded by $a_w^x$, there can be positive trades in which both $v$ and $w$ improve strictly, even in networks without default cost\footnote{Consider three banks $u$, $v$ and $w$, and two edges $(u,v)$ and $(v,w)$ each with liability 1. No bank has external assets. In this network there are no payments, and all assets are 0. Suppose $(u,v)$ is traded completely to $w$ with return 1. Given a return of 1, $v$ gets assets of 1, which are paid back entirely to $w$ via edge $(v,w)$. Thereby we ensure that $w$ can pay the return of 1. A consistent clearing state evolves, in which both $v$ and $w$ improve strictly.}. As such, objectives involving both the assets of $T$ and $w$ are indeed meaningful. 

It is an interesting open problem to generalize this result to networks with default cost. Further interesting open problems include computing fractional trades in settings with fixed haircut rates, in which $\alpha_i$ is fixed and given with the input for each edge $e_i$. Another open problem is whether networks without default cost allow to compute optimal Pareto-positive multi-trades of outgoing edges in polynomial time beyond the special case of excess returns. 

Finally, one could also imagine trades of both incoming and outgoing edges of a bank $v$. Then, however, $v$ might be rewarded for trading outgoing edges for which it is the debtor. Since the creditor usually has much more rights to manipulate the status of a claim, a trade that rewards the debtor seems less plausible (yet, potentially interesting from an algorithmic perspective).

\bibliography{literature}


\end{document}